\documentclass[11pt]{article}

\usepackage{amssymb,amsmath}
\usepackage{ifthen}
\usepackage{color}
\usepackage{xspace}
\usepackage{algorithm}
\usepackage[noend]{algorithmic}
\usepackage{graphicx}

\newcommand{\Eat}[1]{} 

\newcommand{\Set}[1]{\ensuremath{\{ #1 \}}}
\newcommand{\SetCard}[1]{\ensuremath{| #1 |}}
\newcommand{\SpecSet}[2]{\ensuremath{\Set{#1 \mid #2}}}
\newcommand{\Ordinal}[1]{\ensuremath{{#1}^{\rm th}}}
\newcommand{\Ceil}[1]{\ensuremath{\left\lceil{#1}\right\rceil}}

\newcommand\Half[1][]{\ensuremath{\frac{\ifthenelse{\equal{#1}{}}{1}{#1}}{2}}\xspace}

\newcommand{\CCC}{\ensuremath{\mathcal{C}}\xspace}

\newcommand{\HHH}{\ensuremath{\mathcal{H}}\xspace}

\newcommand{\OOO}{\ensuremath{\mathcal{O}}\xspace}

\newcommand{\TTT}{\ensuremath{\mathcal{T}}\xspace}

\newcommand{\XXX}{\ensuremath{\mathcal{X}}\xspace}

\DeclareSymbolFont{AMSb}{U}{msb}{m}{n}
\DeclareMathSymbol{\N}{\mathord}{AMSb}{"4E}
\DeclareMathSymbol{\Z}{\mathord}{AMSb}{"5A}
\DeclareMathSymbol{\R}{\mathord}{AMSb}{"52}

\makeatletter
\newcommand*{\RNum}[1]{\expandafter\@slowromancap\romannumeral #1@}
\makeatother

\newcommand{\InsertAlgorithm}[3]%
{\begin{algorithm}[ht]
\caption{\sc #1}\label{#2}
\begin{algorithmic}[1]
\vspace{0.1cm}
\baselineskip=1.1\baselineskip
#3
\end{algorithmic}\end{algorithm}}

\newcommand{\InsertAlgorithmInOneColumn}[3]%
{\begin{algorithm*}[ht]
\caption{\sc #1}\label{#2}
\begin{algorithmic}[1]
\vspace{0.1cm}
\baselineskip=1.1\baselineskip
#3
\end{algorithmic}\end{algorithm*}}

\def\AlgAssign{\ensuremath{\leftarrow}\xspace}

\definecolor{USCcardinal}{rgb}{0.598,0.000,0.000}
\definecolor{USCgold}{rgb}{0.996,0.797,0.000}

\usepackage[section]{./definitions}

\usepackage[letterpaper,top=1in,left=1in,right=1in,bottom=1in]{geometry}

\begin{document}

\Eat{
\def\element{item\xspace}
\def\elements{items\xspace}
\def\Element{Item\xspace}
\def\Elements{Items\xspace}
}

\def\element{element\xspace}
\def\elements{elements\xspace}
\def\Element{Element\xspace}
\def\Elements{Elements\xspace}

\def\partition{group\xspace}
\def\partitions{groups\xspace}
\def\Partition{Group\xspace}
\def\Partitions{Groups\xspace}

\newcommand{\QuickSort}{\mbox{\textsc{QuickSort}}\xspace}
\newcommand{\InsertionSort}{\mbox{\textsc{InsertionSort}}\xspace}
\def\pathlike{path-like\xspace}

\newcommand{\QuickClustering}{\mbox{\textsc{QuickClustering}}\xspace}
\newcommand{\Merge}{\mbox{\textsc{Merge}}\xspace}
\newcommand{\FindSibling}{\mbox{\textsc{FindSibling}}\xspace}
\newcommand{\InsertionClustering}{\mbox{\textsc{InsertionClustering}}\xspace}
\newcommand{\SimulateVertexQuery}{\mbox{\textsc{SimulateVertexQuery}}\xspace}
\newcommand{\TreeWalk}{\mbox{\textsc{TreeWalk}}\xspace}

\def\binarysearch{Binary Search\xspace}
\def\Binarysearch{Binary Search\xspace}


\def\ElementSym{\ensuremath{x}\xspace}
\def\ElementSymP{{\ensuremath{x'}}\xspace}
\def\ElementSymPP{{\ensuremath{x''}}\xspace}
\def\ElementSymL{{\ensuremath{x_L}}\xspace}
\def\ElementSymR{{\ensuremath{x_R}}\xspace}
\def\AllElements{\XXX}
\def\Err{\ensuremath{\delta}\xspace}
\def\PartitionA{\ensuremath{A}\xspace}
\def\PartitionB{\ensuremath{B}\xspace}
\def\PartitionC{\ensuremath{C}\xspace}
\def\PivotA{\ensuremath{x_\PartitionA}\xspace}
\def\PivotB{\ensuremath{x_\PartitionB}\xspace}
\def\TreeA{\Tree[\PartitionA]}
\def\TreeB{\Tree[\PartitionB]}
\def\TreeC{\Tree[\PartitionC]}
\def\TreeAB{\Tree[\PartitionA\PartitionB]}
\def\NumQuery{\ensuremath{k}\xspace}
\def\Target{\ensuremath{t}\xspace}
\def\Separator{\ensuremath{s}\xspace}
\def\Root{\ensuremath{r}\xspace}
\def\RootP{\ensuremath{r'}\xspace}
\def\OracleTriplet{\ensuremath{\OOO_{\text{triplet}}}\xspace}
\def\ConsistentSet{\ensuremath{S}\xspace}
\def\Null{\ensuremath{\perp}\xspace}
\def\MaxChildren{\ensuremath{d}\xspace}
\def\Dimension{\ensuremath{d}\xspace}
\def\Diameter{\ensuremath{D}\xspace}

\newcommand{\Vertex}[1][]{\ensuremath{%
\ifthenelse{\equal{#1}{}}{v}{v_{#1}}}\xspace}
\newcommand{\VertexP}[1][]{\ensuremath{%
\ifthenelse{\equal{#1}{}}{v'}{v'_{#1}}}\xspace}
\newcommand{\VertexPP}[1][]{\ensuremath{%
\ifthenelse{\equal{#1}{}}{v''}{v''_{#1}}}\xspace}
\newcommand{\VertexBar}[1][]{\ensuremath{%
\ifthenelse{\equal{#1}{}}{\overline{v}}{\overline{v}_{#1}}}\xspace}
\newcommand\Potential[2][]{\ensuremath{\Phi_{#2}%
\ifthenelse{\equal{#1}{}}{}{(#1)}}\xspace}
\newcommand{\Cluster}[1][]{\ensuremath{%
\ifthenelse{\equal{#1}{}}{C}{C_{#1}}}\xspace}
\newcommand{\ClusterP}[1][]{\ensuremath{%
\ifthenelse{\equal{#1}{}}{C'}{C'_{#1}}}\xspace}
\newcommand{\ClusterBar}[1][]{\ensuremath{%
\ifthenelse{\equal{#1}{}}{\overline{C}}{\overline{C}_{#1}}}\xspace}
\newcommand{\Clustering}[1][]{\ensuremath{%
\ifthenelse{\equal{#1}{}}{\HHH}{\HHH[#1]}}\xspace}
\newcommand{\SetFam}[1][]{\ensuremath{%
\ifthenelse{\equal{#1}{}}{\HHH}{\HHH[#1]}}\xspace}
\newcommand{\IndClust}[1]{\ensuremath{\HHH[#1]}\xspace}

\newcommand{\ElSet}[1][]{\ensuremath{%
\ifthenelse{\equal{#1}{}}{X}{X_{#1}}}\xspace}
\newcommand{\ElS}[1][]{\ensuremath{%
\ifthenelse{\equal{#1}{}}{x}{x_{#1}}}\xspace}
\newcommand{\ElSP}[1][]{\ensuremath{%
\ifthenelse{\equal{#1}{}}{x'}{x'_{#1}}}\xspace}
\newcommand{\ElSPP}[1][]{\ensuremath{%
\ifthenelse{\equal{#1}{}}{x''}{x''_{#1}}}\xspace}
\newcommand{\ElSL}{{\ensuremath{\ElS[L]}}\xspace}
\newcommand{\ElSR}{{\ensuremath{\ElS[R]}}\xspace}
\newcommand{\Tree}[1][]{\ensuremath{%
\ifthenelse{\equal{#1}{}}{\TTT}{\TTT_{#1}}}\xspace}
\newcommand{\TreeP}[1][]{\ensuremath{%
\ifthenelse{\equal{#1}{}}{\TTT'}{\TTT'_{#1}}}\xspace}
\newcommand{\Height}[1][]{\ensuremath{%
\ifthenelse{\equal{#1}{}}{h}{h_{#1}}}\xspace}
\newcommand{\TreeOpt}[1][]{\ensuremath{%
\ifthenelse{\equal{#1}{}}{\TTT^{*}}{\TTT^{*}_{#1}}}\xspace}
\newcommand{\TreeI}[2][]{\ensuremath{%
\ifthenelse{\equal{#1}{}}{\TTT[#2]}{\TTT^{#1}[#2]}}\xspace}
\newcommand{\Subtree}[2][]{\ensuremath{%
\ifthenelse{\equal{#1}{}}{\overline{\TTT}^{#2}}{\TTT_{#1}^{#2}}}\xspace}

\newcommand{\Query}[1][]{\ensuremath{%
\ifthenelse{\equal{#1}{}}{q}{q_{#1}}}\xspace}
\newcommand{\Response}[1][]{\ensuremath{%
\ifthenelse{\equal{#1}{}}{r}{r_{#1}}}\xspace}
\newcommand{\Counter}[2][]{\ensuremath{%
\ifthenelse{\equal{#1}{}}{c}{c_{#1}}(#2)}\xspace}
\newcommand{\Dis}[2]{\ensuremath{d(#1, #2)}\xspace}
\newcommand{\QC}[1]{\ensuremath{Q(#1)}\xspace}

\title{Adaptive Hierarchical Clustering Using Ordinal Queries}

\author{%
Ehsan Emamjomeh-Zadeh%
\thanks{%
Department of Computer Science,
University of Southern California,
emamjome@usc.edu} 
\and
David Kempe%
\thanks{%
Department of Computer Science,
University of Southern California,
dkempe@usc.edu}
}

\begin{titlepage}

\maketitle

\begin{abstract}

In many applications of clustering
(for example, ontologies or clusterings of animal or plant species),
hierarchical clusterings are more descriptive
than a flat clustering.
A hierarchical clustering over $n$ \elements
is represented by a rooted binary tree
with $n$ leaves, each corresponding to one \element.
The subtrees rooted at interior nodes capture the clusters.
In this paper, we study active learning
of a hierarchical clustering using only ordinal queries.
An ordinal query consists of a set of three \elements,
and the response to a query reveals the two \elements
(among the three \elements in the query)
which are ``closer'' to each other than to the third one.
We say that \elements \ElS and \ElSP
are closer to each other than \ElSPP
if there exists a cluster containing \ElS and \ElSP,
but not \ElSPP.

When all the query responses are correct,
there is a deterministic algorithm that 
learns the underlying hierarchical clustering
using at most $n \log_2 n$ adaptive ordinal queries.
We generalize this algorithm to be robust in a model in which
each query response is correct independently with probability $p > \Half$,
and adversarially incorrect with probability $1 - p$.
We show that in the presence of noise,
our algorithm outputs the correct hierarchical clustering
with probability at least $1 - \Err$,
using $O(n \log n + n \log(1/\Err))$ adaptive ordinal queries.
For our results, adaptivity is crucial:
we prove that even in the absence of noise,
every non-adaptive algorithm requires $\Omega(n^3)$ ordinal queries
in the worst case.

\end{abstract}

\end{titlepage}



\section{Introduction} \label{sec:introduction}

One of the most useful and frequent computational tasks
in interacting with a set of \elements is to \emph{cluster} them,
inferring groups (\emph{clusters}) of \elements
that are ``similar to each other.''
Clustering is often an essential step in ``learning'' facts
about the items, inferring generative models,
deriving scientific or other insights,
or letting users interact with large data sets meaningfully.
The wide range of applications has led to an enormous amount of research
spanning many disciplines, driven by different applications,
different notions of ``similarity,'' different constraints
on allowable groups of \elements,
and different models of user interaction,
to name just a few.

Two particularly common restrictions
on the allowable groups are \emph{flat clustering},
in which each item belongs to exactly (or at most) one cluster,
and \emph{hierarchical clustering},
in which the groups form a nested structure.
Both restrictions have myriad natural applications.
Among the prototypical applications of hierarchical clusterings 
are hierarchies of animal or plant species
\cite{felsenstein:2004:phylogenies},
ontologies of words or other entities,
or social networks, viewed at different levels of granularity.
Hierarchical clusterings are particularly well suited
for human data exploration, in that they allow a person
to zoom in/out to understand the \elements and their relationships
at different levels. In this paper, we are interested in
inferring a hierarchical clustering
on a set \AllElements of $n$ \elements.

A second key question is how the similarities between \elements
are defined and obtained.
In much of the work on clustering, the similarities are derived
from observable features of the \elements
(such as their position in an observable metric space),
and finding a clustering is cast as
an optimization problem for a suitably chosen objective function.
Such an approach frequently identifies clusters quite similar
--- but not identical --- to the ground truth.
An alternative approach is to interact with human users,
and obtain input to guide the search for the true clustering.
These interactions are typically expensive;
thus, minimizing the number of interactions with human users
is often a primary goal in designing algorithms. 
When interacting with humans, it is also important to keep in mind
that even when there is an objective meaning to numerical similarity values,
humans are notoriously bad at estimating such numerical quantities.
This has led to a search for suitable models of interaction,
advocated recently in
\cite{balcan-blum:2008:split-merge,%
awasthi-zadeh:2010:supervised-clustering,%
awasthi-balcan-voevodski:2017:local-algorithm-journal,%
vikram-dasgupta:2016:interactive-hierarchical-clustering},
in which human users are only given queries
they may reasonably be in a position to answer.

Inspired by the equivalence query model for learning a classifier
\cite{angluin:1988:queries-concept}, 
Balcan and Blum~\cite{balcan-blum:2008:split-merge}
and Awasthi et al.~\cite{awasthi-zadeh:2010:supervised-clustering,%
awasthi-balcan-voevodski:2017:local-algorithm-journal}
defined a model of interaction (specifically for flat clustering)
in which an algorithm repeatedly proposes a clustering,
and a user can request corrections, such as splitting or merging clusters.
In this model, the user has a lot of leeway in determining
what feedback the algorithm receives.
The \emph{ordinal query} model
(\cite{
schultz-joachims:2003:relative-comparisons,%
jain-jamieson-nowak:2016:ordinal-embedding,%
vikram-dasgupta:2016:interactive-hierarchical-clustering,%
kendall-gibbons:1990:rank-correlation}, for example)
is more reminiscent of traditional ``active learning'' models.
Here, a user is given queries comprising triplets of \elements,
and asked to identify the pair from each triplet that is closest.
For example, a user may be asked which pair out of
\{lion, dog, shark\} is most similar.

One well-known application of this model
is to discover the phylogeny trees of a set of species.
A phylogeny tree over a set of species
is a rooted binary tree whose leaves correspond to the species,
while internal nodes capture the evolutionary history.
It is often assumed that there is a black box
which can answer ordinal queries on a set of three DNA sequences
and infer which two are biologically closer to each other;
under this assumption,
the problem of finding the underlying phylogeny
is similar to the problem we study in this paper~%
\cite{kannan-lawler-warnow:1996:phylogeny-triplet,%
brown-truszkowski:2011:phylogeny-quartet,%
brown-truszkowski:2011:phylogeny-quartet-practical}.
Other topics such as social computing and a computational lens
on political and other decision making processes have led to
increasing interest in ordinal query responses as well.
As a result, ordinal queries have
been analyzed in the contexts of clustering or inferring distances
\cite{
schultz-joachims:2003:relative-comparisons,%
agarwal-wills-cayton-lanckriet-kriegman-belongie:2007:embedding,%
tamuz-liu-belongie-shamir-kalai:2011:learning-kernel,%
jain-jamieson-nowak:2016:ordinal-embedding,%
vikram-dasgupta:2016:interactive-hierarchical-clustering},
in social choice scenarios such as voting
\cite{anshelevich-bhardwaj-postl:2015:social-choice,%
anshelevich:2016:ordinal},
causal inference~\cite{pearl-tarsi:1986:strucuting},
and political and group decision making
\cite{goel-lee:2012:triadic,%
goel-lee:2014:small-group-interaction}.
\emph{Learning of a hierarchical clustering
from ordinal queries is the topic of the present work.}

We assume that there is
an unknown ground truth hierarchical clustering to be learned.
It is represented by an (unknown) tree \TreeOpt
whose $n$ leaves are the \elements.
For each node \Vertex of the tree,
the leaves in the subtree rooted at \Vertex
form a cluster in the hierarchical clustering.
For most of the paper
(except Section~\ref{sec:generalized-hierarchical-clustering}),
we assume that \TreeOpt is binary.
This implies that for each triplet \Set{\ElS, \ElSP, \ElSPP}
of leaves/\elements,
there is a unique pair that shares a closer common ancestor
than any other pair:\footnote{This does not hold in trees of higher degree.}
this pair is the correct \emph{response} to the query \Set{\ElS, \ElSP, \ElSPP}.
The goal is to learn \TreeOpt (or an isomorphic tree)
using few ordinal queries.

For some of this paper,
we assume that the answers to all ordinal queries are correct.
However, in particular in light of the motivation
of interaction with humans,
it is natural to assume that responses may be noisy.
In the \emph{independent noise mode}
(which we use in Section~\ref{sec:insertion-noisy}),
each response is correct independently with
probability $p > \Half$, and adversarially incorrect otherwise.

It is not difficult to show (see Section~\ref{sec:ordinal-def})
that even without incorrect responses,
when the ordinal queries are non-adaptive
(i.e., have to be chosen ahead of time),
$\Theta(n^3)$ queries are necessary and sufficient.
We therefore focus on adaptive queries.

For adaptive queries, when \TreeOpt has height $O(\log n)$,
there is a simple top-down algorithm to learn \TreeOpt
using $O(n \log n)$ adaptive ordinal queries
(with no incorrect responses).
This algorithm can be considered as a special case of an
algorithm of Pearl and Tarsi~\cite{pearl-tarsi:1986:strucuting}.
Their algorithm was designed for a different context, namely,
discovering underlying causal trees over visible random variables.
However, as a preliminary step, they studied a model essentially
identical to the ordinal query model.
Their algorithm learns the underlying tree (of any height)
using $O(n \log n)$ queries.
In Section~\ref{sec:insertion-clean}, we will review their algorithm
as a point of departure for our noise-tolerant algorithm in
Section~\ref{sec:insertion-noisy}.
A matching lower bound of $\Omega(n \log n)$
follows easily from a counting argument
like the one for sorting (see Section~\ref{sec:quick-sort}).

\subsection*{Our Results and Techniques}

The analogy with sorting is also useful in the design of algorithms
for inferring a hierarchy.
We begin (in Sections~\ref{sec:quick-sort} and \ref{sec:insertion-clean})
by studying the model without incorrect responses.

Our first (and simplest) algorithm (in Section~\ref{sec:quick-sort})
is modeled after randomized \QuickSort.
It picks \emph{two} random pivot \elements \PivotA, \PivotB,
and partitions the \elements into those closest to \PivotA,
those closest to \PivotB,
and those which are further from both pivots
than the two pivots are from each other.
We prove that with constant probability,
each of the three sets constitutes
only a constant fraction of the original set.
Then, we present an algorithm that, 
once the hierarchies for all three sets have been (recursively) computed,
can merge them using a linear number of queries.
As a result, the algorithm uses $O(n \log n)$ ordinal queries.

In preparation for the noise-tolerant algorithm in
Section~\ref{sec:insertion-noisy},
in Section~\ref{sec:insertion-clean},
we review\footnote{We would like to thank Daniel Hsu for
pointing us to the article \cite{pearl-tarsi:1986:strucuting}.}
the algorithm of Pearl and Tarsi
\cite{pearl-tarsi:1986:strucuting},
which can be considered as a variant of \InsertionSort. 
The main observation
guiding the algorithm in Section~\ref{sec:insertion-clean}
is that, given the true hierarchy over $i$ \elements,
even when the tree is arbitrarily unbalanced, 
a new \element can be inserted using at most
$1 + \log_2 i$ ordinal queries,
giving an overall query complexity of $n \log_2 n$.
The algorithm to insert an \element is a close variant
of the generalization of \binarysearch to trees or arbitrary graphs
\cite{onak-parys:2006:tree-vertex-linear,%
mozes-onak-weimann:2008:tree-edge-linear,%
2016:binary-search}.
While the work on \binarysearch in trees/graphs
assumes a different query model called \emph{vertex queries}
(see Section~\ref{sec:vertex-queries}),
the algorithm is easily adapted.

Next, we turn our attention to the noisy model.
It is easy to see that any algorithm
that learns the true hierarchy in the absence of noise
can be made robust to noise by repeating each ordinal query
$\Theta(\log n)$ times and using a majority output.
Then, Hoeffding's inequality and union bounds guarantee that
the algorithm succeeds with high probability.
However, this approach increases the
number of queries by a factor of $\Theta(\log n)$.
In Section~\ref{sec:insertion-noisy},
we show how to avoid the worse dependence on $n$
for the \InsertionSort-based algorithm,
and only increase the number of queries by a constant factor
depending solely on the correctness probability $p$.

The robustness is achieved by replacing
the \binarysearch subroutine with a robust version.
While robust \binarysearch algorithms
in a vertex query model in the presence of noise
had been given in work of
Feige et al.~\cite{feige-raghavan-peleg-upfal:1994:noisy}
and Emamjomeh-Zadeh et al.~\cite{2016:binary-search},
neither algorithm itself has strong enough guarantees for our purposes:
the algorithm of \cite{feige-raghavan-peleg-upfal:1994:noisy} 
requires a number of queries linear in the diameter of the tree,
while the algorithm of \cite{2016:binary-search} has a dependency on
the desired error guarantee that would lead to a query complexity of
$\Theta(\log ^2 n)$ in our context.
Our main contribution in Section~\ref{sec:insertion-noisy}
is a hybrid algorithm combining the first stage of the algorithm
from \cite{2016:binary-search} with the algorithm of
\cite{feige-raghavan-peleg-upfal:1994:noisy} to obtain a number of
queries that is logarithmic in $n$.
Specifically, we show how, given the tree over a subset of \elements,
to correctly insert a new \element into the tree
with probability $1 - \Err$, using $O(\log n + \log (1/\Err))$ queries.
This leads to an algorithm that learns the correct hierarchy \TreeOpt
with probability at least $1-\Err$ using
$O(n (\log n + \log(1/\Err)))$ ordinal queries.

\subsection*{Related Work}

The algorithmic problem of inferring the evolutionary history
(known as \emph{phylogeny}) of a set of animal or plant species
(see, e.g., \cite{aho-sagiv-szymanski-ullman:1981:inferring-tree,%
kannan-lawler-warnow:1996:phylogeny-triplet,%
erdos-steel-szekely-warnow:1999:phylogeny-quartet-a,%
brodal-fagerberg-pedersen-ostlin:2001:phylogeny-higher-degree,%
felsenstein:2004:phylogenies,%
gronau-moran-snir:2008:phylogeny-short-edges,%
wu-kao-lin-you:2008:pylogeny-noisy,%
truszkowski:2013:thesis})
is closely related to the problem of learning an underlying hierarchy
over a set of elements.
A (rooted) phylogeny over $n$ species is defined as a rooted binary tree
with $n$ leaves corresponding to the species.
The internal nodes represent
past extinct ancestors or hypothesized ones.
Given (the DNA sequence of) three species,
an ``experiment'' determines which two
are (biologically) closer to each other than to the third one~%
\cite{kannan-lawler-warnow:1996:phylogeny-triplet,%
wu-kao-lin-you:2008:pylogeny-noisy}.
Hence, discovering the phylogeny using a minimum number of experiments
is equivalent to the problem we study here.
Alternatively, the phylogeny of $n$ species
is sometimes defined as an unrooted tree with $n$ leaves
in which each internal node has degree exactly $3$.
In this model, queries are in the form of \emph{quartet} queries:
each query consists of $4$ species,
and the response to it indicates which pair can be separated
from the other pair by removing one edge in the underlying phylogeny.
Although there may not be a direct (algorithmic) reduction
between the two models, most ideas and algorithms that are proposed
in the literature for one of the models
directly carry over for the other one as well.

Assuming that all the query responses are correct,
Kannan, Lawler and Warnow~\cite{kannan-lawler-warnow:1996:phylogeny-triplet}
discuss the trade-off between the query complexity of the algorithm
and its running time. 
All three algorithms they present
achieve an $O(n \log n)$ upper bound on the number of queries;
they show that by increasing the query complexity
by a constant factor, one can achieve faster algorithms in return.
Their algorithms are similar to
the algorithm we present in Section~\ref{sec:insertion-clean}.
In this paper, we only focus on the query complexity
and do not analyze the running time of our algorithms
(except for noting that they are obviously polynomial-time).

The algorithms which are used in computational biology
to answer (triplet or quartet) queries
are not fully reliable;
hence, it is necessary to account for noise in the 
query responses.
Wu et al.~\cite{wu-kao-lin-you:2008:pylogeny-noisy}
studied the problem of recovering an underlying unrooted phylogeny
using quartet queries under an independent noise model.
They presented an algorithm that achieves an $O(n^3 \log n)$ upper
bound on the number of quartet queries.

Brown and Truszkowski%
~\cite{brown-truszkowski:2011:phylogeny-quartet,%
brown-truszkowski:2011:phylogeny-quartet-practical}
improved this result by designing a randomized algorithm
with a query complexity of $O(n \log n)$,
both in expectation and with high probability.
The running time of their algorithm
is also $O(n \log n)$ in expectation and with high probability.
Their algorithm is very similar%
\footnote{The original version of our paper was submitted and accepted
before we learned about
\cite{brown-truszkowski:2011:phylogeny-quartet} and
\cite{brown-truszkowski:2011:phylogeny-quartet-practical}.
We thank the authors for informing us about these publications.}
to our algorithm in Section~\ref{sec:insertion-noisy}:
They insert the elements one by one into an existing tree,
using an idea inspired by \cite{feige-raghavan-peleg-upfal:1994:noisy},
and guarantee (as we do) that each insertion is correct
with high enough probability.
They prove that if the elements are
randomly permuted at the beginning of the algorithm
(i.e., the order in which they are inserted is uniformly random),
then each insertion takes $O(\log n)$ queries
both in expectation and with high probability.
Our algorithm, in contrast, is deterministic
and ensures that each insertion uses
no more than $O(\log n)$ queries.
\cite{brown-truszkowski:2011:phylogeny-quartet}
assume that noise is \emph{permanent},
which means that if the algorithm asks the same query multiple times,
the responses are the same.
This assumption is necessary for their application,
because query responses are obtained
based on fixed DNA sequences,
using fixed algorithms in computational biology.
The algorithm of \cite{brown-truszkowski:2011:phylogeny-quartet}
is designed to recover the exact phylogeny
without ever querying the same quartet more than once.
In our model, on the other hand, each query response is incorrect
with probability $1 - p$, independently of all previous queries
(even if the same query has been asked before).
When queries are triplet (rather than quartet) queries,
if the noise is permanent, it is impossible
to recover the exact structure of the underlying rooted phylogeny
(or hierarchical clustering) with high probability.
As pointed out by~%
\cite{brown-truszkowski:2011:phylogeny-quartet,%
brown-truszkowski:2011:phylogeny-quartet-practical},
the independent noise model is not, in general,
a natural model for inferring phylogenetic trees.
The reason is that query responses are based on the DNA
sequences of the species, and hence, they are not independent.

Vikram and Dasgupta~\cite{vikram-dasgupta:2016:interactive-hierarchical-clustering}
addressed the problem of incorporating hard ordinal constraints
into algorithms for hierarchical clusterings.
However, contrary to our model, they assume that
the algorithm is provided information regarding the geometry of the data,
and the hard constraints given by the user are
additional information to improve the accuracy of the outcome.

A hierarchical clustering over $n$ \elements
can be seen as a metric (or in fact, ultra-metric)
over the \elements.\footnote{For instance,
let the distance between every pair of \elements
be the size of the smallest cluster in the hierarchy
containing both of them.}
In this sense, learning a hierarchical clustering is related
to a recent line of work (mostly in the machine learning community)
\cite{agarwal-wills-cayton-lanckriet-kriegman-belongie:2007:embedding,%
jamieson-nowak:2011:embedding,%
jamieson-nowak:2011:active-ranking,%
McFee-lanckriet:2011:multi-model-similarity,%
tamuz-liu-belongie-shamir-kalai:2011:learning-kernel,%
VanDerMaaten-weinberger:2012:stochastic-triplet-embedding,%
kleindessner-vonLuxburg:2014:embedding-uniqueness,%
jain-jamieson-nowak:2016:ordinal-embedding}
on constructing a geometric interpretation of data
based on ordinal information.
In particular, the ``non-metric multidimensional scaling'' problem
\cite{agarwal-wills-cayton-lanckriet-kriegman-belongie:2007:embedding,%
jamieson-nowak:2011:embedding,%
kleindessner-vonLuxburg:2014:embedding-uniqueness,%
jain-jamieson-nowak:2016:ordinal-embedding} is formulated as
finding an embedding of the \elements
into a (low-dimensional) Euclidean space,
such that the distances between the points
satisfy a given set of ordinal constraints.
Many of the algorithms on the non-metric multidimensional scaling
problem are in the non-adaptive model (e.g.,
\cite{agarwal-wills-cayton-lanckriet-kriegman-belongie:2007:embedding,%
jain-jamieson-nowak:2016:ordinal-embedding}).
Jamieson and Nowak~\cite{jamieson-nowak:2011:embedding}
point out that an adaptive algorithm may be able to
find an embedding of the \elements into the low-dimensional space
using only a few ordinal queries,
in a way that the embedding satisfies
all the $\Theta(n^3)$ possible queries.
They proved a lower bound of $\Omega(\Dimension n \log n)$
where \Dimension is the dimension,
and they conjectured that this lower bound is tight.

In the context of the \emph{multinomial logit model},
a finite set of choices \CCC satisfies the
``Independence of Irrelevant Alternatives'' (IIA) property%
~\cite{McFadden-tye-train:1977:multinomial-logit}
if for every two choices $c_1, c_2 \in \CCC$,
the ratio of the probabilities of being chosen
is a constant, regardless of which other choices (from \CCC)
are available or unavailable. That is,
for every two subsets of choices $\CCC', \CCC'' \subseteq \CCC$
with $\Set{c_1, c_2} \subseteq \CCC', \CCC''$,
\begin{align*}
\frac{\Pr[c_1 \text{ is chosen } | \; \CCC']}
{\Pr[c_2 \text{ is chosen } | \; \CCC']} =
\frac{\Pr[c_1 \text{ is chosen } | \; \CCC'']}
{\Pr[c_2 \text{ is chosen } | \; \CCC'']}.
\end{align*}
In many applications where this strong assumption is not met,
the \emph{nested logistic regression} model%
~\cite{benson-kumar-tomkins:2016:nested-IIA} is used.
In the nested logistic regression model,
the choices are the leaves of an underlying (usually unknown) rooted tree
(which need not be binary).
A pair of choices $c_1, c_2$ is independent
of a third choice $c_3$
if there is a (rooted) subtree containing $c_1$ and $c_2$, but not $c_3$.
Using the terminology of this paper,
the pair $c_1, c_2$ is independent of $c_3$ if
$c_1$ and $c_2$ are closer to each other than to $c_3$.

In a recent paper,
Benson et al.~\cite{benson-kumar-tomkins:2016:nested-IIA}
considered the problem of algorithmically recovering
the underlying tree structure from an observed set of choices.
In particular, they showed that using $O(n^2)$ adaptive ordinal queries,
one can recover the tree.
In fact, their algorithm is a deterministic version
of our \QuickSort-based algorithm from Section~\ref{sec:quick-sort}.
As we discuss in Section~\ref{sec:generalized-hierarchical-clustering},
for trees of very high degree,
the performance of the generalization of our algorithm
deteriorates to $\Theta(n^2)$,
matching the guarantee of
\cite{benson-kumar-tomkins:2016:nested-IIA}.
However, for low-degree trees, and in particular,
for binary trees, our algorithm is provably more efficient.

Yet another application of the algorithmic problem
is \emph{inferring underlying causal structures}.
Pearl and Tarsi~\cite{pearl-tarsi:1986:strucuting}
introduced a tree-like model
where the leaves are observable random variables
and the internal nodes are hidden causes.
They considered a statistical test over the observable variables
(i.e., leaves) which is equivalent to the notation of ordinal query
we discuss in this article.
They showed that using $O(n \log n)$ statistical tests,
one can discover the underlying causal tree
if the degrees are bounded by a constant.
We discuss their algorithm in Section~\ref{sec:insertion-clean}
as the foundation of our main result in Section~\ref{sec:insertion-noisy}.

Feige et al.~\cite{feige-raghavan-peleg-upfal:1994:noisy}
studied several problems under the independent noise model.
In particular, they showed that for the classical \binarysearch problem,
if comparisons are noisy
(i.e., each comparison is correct independently with probability $p > \Half$
and flipped otherwise),
there exists an algorithm that uses $O(\log n)$ queries
and succeeds with high probability.
Ben-Or and Hassidim~\cite{BenOr-hassidim:2008:noisy-binary-search}
improved this result and obtained
the information-theoretically optimal query complexity. 
Braverman and Mossel~\cite{braverman-mossel:2008:noisy-sorting}
studied the sorting problem under an independent noise model
in a non-adaptive setting.
They proposed an algorithm that,
given the noisy comparisons for every pair of \elements,
finds the maximum-likelihood permutation.
They also showed that the maximum-likelihood permutation
is ``close'' to the ground truth.


\section{Definitions and Preliminaries} \label{sec:preliminaries}

\subsection{Hierarchical Clusterings}

Let \AllElements be a set of $n \geq 2$ \elements.
A \emph{hierarchical clustering}
of \AllElements can
be equivalently characterized in two ways:
\begin{enumerate}
\item A laminar family \SetFam of subsets\footnote{%
Recall that a family of sets is
\emph{laminar} if $A, B \in \SetFam$ and $A \cap B \neq \emptyset$
implies that $A \subseteq B$ or $B \subseteq A$.}
(called \emph{clusters}) of \AllElements,
with the properties that $\emptyset \notin \SetFam$, 
$\AllElements \in \SetFam$ and
$\Set{\ElS} \in \SetFam$ for all $\ElS \in \AllElements$.
\item A rooted tree \Tree
whose leaves are the \elements $\ElS \in \AllElements$.
\end{enumerate}

The equivalence is obvious:
the clusters in \SetFam exactly correspond to vertices \Vertex in \Tree,
or more specifically, to the set of leaves in the subtree rooted at \Vertex.
We call the set of all clusters corresponding to rooted subtrees of \Tree
the \emph{hierarchical clustering corresponding to \Tree},
and denote it by \IndClust{\Tree}.
When there is no risk of confusion,
we also refer to \Tree itself as a hierarchical clustering.

We primarily use the representation in terms of the tree \Tree.
In keeping with much of the prior literature,
(\cite{dasgupta:2016:hierarchy-cost,%
clauset-moore-newman:2008:link-prediction,%
felsenstein:2004:phylogenies}, for example)
we focus
(except in Section~\ref{sec:generalized-hierarchical-clustering})
on the case when \Tree is binary.
In terms of the representation \SetFam, this means that for each
$\ElSet \in \SetFam$ with $\SetCard{\ElSet} \geq 2$,
there exists a set $A \subsetneq \ElSet$
with $A, \ElSet \setminus A \in \SetFam$.
We assume without loss of generality that \Tree has no internal node
with one child,
since such a node could be removed without changing
the corresponding family of sets.

There is an unknown ground truth hierarchical clustering \TreeOpt,
which an algorithm is to recover using ordinal queries.
The algorithm succeeds if it finds a rooted binary tree \Tree
with $\IndClust{\Tree} = \IndClust{\TreeOpt}$;
this corresponds to the rooted binary trees \Tree and \TreeOpt being isomorphic
(when the ordering of left vs.~right subtrees is immaterial).
In this case, we also write $\Tree \equiv \TreeOpt$.

For a set $S \subseteq \AllElements$ with
$\SetCard{S} \geq 2$,
the \emph{induced hierarchy} \SetFam[S] is defined in the obvious way:
$\SetFam[S] = \SpecSet{\ElSet \cap S}{\ElSet \in \SetFam}$.
We define the \emph{induced tree} \TreeI{S} as follows:
remove from \Tree all leaves not in $S$,
then repeatedly remove all internal nodes without children,
and contract all nodes with a single child.\footnote{%
If the root has one child, it is removed,
and its child becomes the root of the tree.}

\begin{proposition} \label{prop:induced}
\SetFam[S] is the hierarchy corresponding to the tree induced by $S$,
i.e., $\SetFam[S] = \IndClust{\TreeI{S}}$.
\end{proposition}

In reasoning about trees, we use the following notation.
For an \emph{internal} vertex \Vertex,
let $\Subtree[L]{\Vertex}, \Subtree[R]{\Vertex}$
be the subtrees rooted at \Vertex's left and right children, and let
$\Subtree{\Vertex} = \Tree \setminus (\Subtree[L]{\Vertex} \cup \Subtree[R]{\Vertex})$
be the set of all other vertices (so $\Vertex \in \Subtree{\Vertex}$).
As is standard for \emph{rooted} trees, we consider
the \emph{degree} of a node to be its number of children,
i.e., we do not count its parent.

\subsection{Ordinal Queries} \label{sec:ordinal-def}

An \emph{ordinal} query is a set of three distinct \elements
$\Set{\ElS, \ElSP, \ElSPP} \subseteq \AllElements$.
The response to the query reveals which two \elements
are ``closest'' to each other.
We say that \ElS and \ElSP are closer to each other than to \ElSPP
(with respect to \Tree) 
if there exists a cluster (subtree rooted at an internal node)
in \Tree which contains \ElS and \ElSP, but not \ElSPP.
In other words, for every set of three \elements,
the two \elements which have the lowest common ancestor are closest.
Because \Tree is a binary tree,
such a cluster --- and hence a valid response --- always exists
and is unique.

\begin{proposition} \label{prop:tree-unique}
Two hierarchical clusterings \Tree, \TreeP of the same \elements \AllElements
satisfy $\Tree \equiv \TreeP$ if and only if
for every triplet $\Set{\ElS, \ElSP, \ElSPP}$,
the query responses with respect to \Tree and \TreeP are the same.
\end{proposition}

Thus, using at most $\binom{n}{3} = \Theta(n^3)$ ordinal queries,
an algorithm can uniquely recover the underlying hierarchy \TreeOpt.
(The actual algorithm is a simple bottom-up algorithm.)
For non-adaptive algorithms, the following easy theorem
(proved in Appendix~\ref{sec:non-adaptive-lower-bound})
gives a matching lower bound.

\begin{theorem} \label{thm:hardness-non-adaptive}
Every non-adaptive algorithm requires $\Omega(n^3)$ ordinal queries
to learn the hierarchical clustering over $n$ \elements
(even in the absence of noise).
\end{theorem}

Our goal is to reduce the number of queries to $O(n \log n)$
using adaptive ordinal queries.
The algorithm has access to an oracle (for example, a human user) that
will answer ordinal queries based on the (otherwise unknown)
ground truth hierarchy \TreeOpt.
For the next two sections, we assume that these answers are correct;
in Section~\ref{sec:insertion-noisy},
we turn our attention to the  \emph{independent noisy} model.
In this model, there is a known constant $p > \Half$
such that each query response is correct \emph{independently}
with probability $p$,
and adversarially incorrect\footnote{Of course, the adversary could
give a correct answer if an algorithm were to rely on it
always answering incorrectly.} with probability $1 - p$.
Because the query responses are independent,
if an algorithm queries the same triplet multiple times,
it may get different responses.

In several of our algorithms, we will try to locate where within an
existing (partial) hierarchy \Tree another \element \ElS
should be inserted.
In doing so, it is particularly useful to ``query'' an internal
vertex \Vertex, and learn whether \ElS should be in the left subtree
of \Vertex, the right subtree of \Vertex, or neither.
This can be accomplished as follows.
Let $\ElSL \in \Subtree[L]{\Vertex}, \ElSR \in \Subtree[R]{\Vertex}$
be arbitrary \elements in the left and right subtrees of \Vertex.
(\ElSL, \ElSR exist because \Vertex is an internal vertex,
and thus has exactly two children.)
We call the query $\Set{\ElSL, \ElSR, \ElS}$ an 
\emph{ordinal query for \ElS with pivot \Vertex}.
Its possible outcomes can be interpreted as follows.

\begin{proposition} \label{prop:pivot-query}
If the response to \Set{\ElSL, \ElSR, \ElS} is that
\ElSL and \ElS are closest,
then \ElS's correct location is in \Subtree[L]{\Vertex}
(the left subtree of \Vertex).
If the response to $\Set{\ElSL, \ElSR, \ElS}$ is that
\ElSR and \ElS are closest,
then \ElS's correct location is in \Subtree[R]{\Vertex}.
If \ElSR and \ElSL are closest, then \ElS's correct location is not
in the subtree rooted at \Vertex, i.e., \ElS belongs in \Subtree{\Vertex}.
\end{proposition}

Ordinal queries with pivot \Vertex are only defined when \Vertex is an
internal node, and we will only use them in that case.
When notationally convenient, we will consider the left/right child
of \Vertex itself as the response to an ordinal query with pivot \Vertex.

\subsection{Vertex Queries} \label{sec:vertex-queries}

Ordinal queries with pivot \Vertex provide ``directional'' information
regarding the location of \ElS with respect to \Vertex.
As such, they are similar to (though subtly different from)
\emph{vertex queries} defined in past work
\cite{onak-parys:2006:tree-vertex-linear,2016:binary-search}.
In the vertex query model, the goal is to identify an a priori unknown
\emph{target node} \Target in a known
(undirected, not necessarily rooted) tree \Tree,
by repeatedly querying nodes \Vertex of \Tree.
In response to a vertex query, the algorithm is told that
$\Target = \Vertex$, or otherwise is given the unique neighbor that is
closer to the target.

Vertex queries are subtly more powerful than ordinal queries with a pivot,
because they can identify \Vertex as the target.
We will show later how to simulate (for our purposes)
vertex queries with a constant number of ordinal queries with pivots.
In fact, our analysis of the algorithm inspired by \InsertionSort is
based on ideas from prior work in the vertex query model,
and the analysis for the independent noisy model explicitly uses a
reduction from the ordinal query model to the vertex query model.


\section{A \QuickSort-like Randomized Algorithm}
\label{sec:quick-sort}

In this section, we propose a simple randomized algorithm
(Algorithm~\ref{alg:quick-sort}),
reminiscent of \QuickSort.
The algorithm uses $O(n \log n)$ ordinal queries in expectation
to learn the hierarchy over $n$ \elements.
Algorithm~\ref{alg:quick-sort} 
randomly partitions the \elements into three \partitions
\PartitionA, \PartitionB, and \PartitionC
(lines~\ref{line:partitioning-begin}--\ref{line:partitioning-end}),
until all of the partitions are small enough.
Once they are, the algorithm recursively determines the hierarchy for
each partition, and finally merges them.

\InsertAlgorithmInOneColumn{\QuickClustering$(\ElSet)$}{alg:quick-sort}{%
\IF{$\SetCard{\ElSet} \leq 2$}
	\RETURN{the obvious hierarchy.}
\ENDIF
\REPEAT
	\STATE{Let $(\PivotA, \PivotB)$
	be a pair of distinct \elements chosen uniformly at random.}
	\label{line:partitioning-begin}
	\FOR{each $\ElS \in \ElSet \setminus \Set{\PivotA, \PivotB}$}
		\STATE{Make an ordinal query of \Set{\PivotA, \PivotB, \ElS}.}
	\ENDFOR
	\STATE{Let $\PartitionA \AlgAssign \Set{\PivotA} \cup
	\SpecSet{\ElS \in \ElSet}{\text{\ElS and \PivotA are closer to each other
	than to 	\PivotB}}$.}
	\STATE{Let $\PartitionB \AlgAssign \Set{\PivotB} \cup
	\SpecSet{\ElS \in \ElSet}{\text{\ElS and \PivotB are closer to each other
	than to \PivotA}}$.}
	\STATE{Let $\PartitionC \AlgAssign 
  	\SpecSet{\ElS \in \ElSet}{\text{\PivotA and \PivotB are closer to each other
  	than to \ElS}}$.}
\UNTIL{$\max(\SetCard{\PartitionA}, \SetCard{\PartitionB}, \SetCard{\PartitionC})
\leq \frac{15}{16} \cdot \SetCard{\ElSet}$.}
\label{line:partitioning-end}
\STATE{$\TreeA \AlgAssign \QuickClustering(\PartitionA)$.}
\STATE{$\TreeB \AlgAssign \QuickClustering(\PartitionB)$.}
\STATE{$\TreeC \AlgAssign \QuickClustering(\PartitionC)$.}
\RETURN{\Merge $(\TreeA, \TreeB, \TreeC)$.}
}

The partitioning is akin to the partitioning stage in \QuickSort: 
it draws a pair of distinct \emph{pivot \elements}
$\PivotA, \PivotB \in \ElSet$, uniformly at random.
It then partitions the \elements of \ElSet into three sets.
These sets can be characterized as follows:
Let \VertexBar be the lowest common ancestor of \PivotA, \PivotB
in the ground truth hierarchy \TreeOpt,
with children \Vertex[A], \Vertex[B].
Then, \PartitionA (\PartitionB) consists of all \elements in the
subtree of \TreeOpt rooted at \Vertex[A] (\Vertex[B]),
while \PartitionC consists of all remaining \elements:
exactly the \elements outside of the subtree of \TreeOpt rooted at \VertexBar.
Therefore, after the recursive calls,
\TreeA and \TreeB are the subtrees of \TreeOpt rooted at \Vertex[A]
and \Vertex[B], while \TreeC is the result of removing \VertexBar
and its subtree from \TreeOpt (and contracting its parent).

The role of the \Merge function (Algorithm~\ref{alg:merge}) is then to
insert \VertexBar and its subtree in the correct place in \TreeC;
the primary task here is to correctly identify the sibling \Vertex of \VertexBar.
The algorithm does this by starting at the root and repeatedly using
ordinal queries for a representative \element \ElS of
$\TreeA \cup \TreeB$ with \Vertex as a pivot,
and following the direction returned by the ordinal query.
That inserting \VertexBar as a sibling of \Vertex
(with a new parent \VertexP) is correct follows from
Proposition~\ref{prop:induced}.

\InsertAlgorithmInOneColumn{\Merge $(\TreeA, \TreeB, \TreeC)$}{alg:merge}{%
\STATE{Add a vertex \VertexBar whose children are the roots of \TreeA and \TreeB.}
\STATE{Let \ElS be an arbitrary \element (i.e., leaf) in \TreeA or \TreeB.}
\STATE{Let \Vertex be the root of \TreeC.}
\WHILE{\Vertex is not a leaf, and the ordinal query for \ElS
with pivot \Vertex does not return \Subtree{\Vertex}}
	\STATE{Update \Vertex to the child of \Vertex
returned by that ordinal query for \ElS}
\ENDWHILE
\STATE{Insert a new vertex \VertexP as parent of \Vertex, and make
\VertexBar its other child.}
}

We next analyze the expected number of ordinal queries required
by the \QuickClustering algorithm.
The number of queries required for \Merge is at most\footnote{%
In fact, using the techniques from Section~\ref{sec:insertion-clean},
the bound could be improved to $1 + \log_2 \SetCard{\ElSet}$;
but this improvement is drowned out by the number of queries
for the partitioning.}
$\SetCard{\ElSet}$,
and the number of queries for each iteration of the partitioning loop
(lines \ref{line:partitioning-begin}--\ref{line:partitioning-end}
in Algorithm~\ref{alg:quick-sort})
is $\SetCard{\ElSet} - 2$.
The key observation is that in expectation,
the loop is executed only a constant number of times
before succeeding, captured by Lemma~\ref{lem:separator}.

\begin{lemma} \label{lem:separator}
The probability that in line~\ref{line:partitioning-end}
of Algorithm~\ref{alg:quick-sort},
at least one $\PartitionA, \PartitionB, \PartitionC$
is larger than $\frac{15}{16} \cdot \SetCard{\ElSet}$
is at most $\frac{125}{128}$.
\end{lemma}

\begin{proof}
We first prove that there exist clusters
\Cluster and \ClusterP with
$\ClusterP \subseteq \Cluster$ in \IndClust{\Tree}
such that $\frac{n}{16} < \SetCard{\ClusterP} \leq \frac{n}{8}$
and $\frac{n}{4} < \SetCard{\Cluster} \leq \frac{n}{2}$.
This follows because \Tree is binary:
starting from $\Cluster = \AllElements$,
we can repeatedly consider a bipartition of
\Cluster into two disjoint subclusters.
While \Cluster has more than $\frac{n}{2}$ \elements,
we can replace it by the larger of the two subclusters,
which must have at least half the size of \Cluster.
We end up with a set \Cluster of the required size.
Continuing from $\ClusterP = \Cluster$ in the same way until
\ClusterP has size at most $\frac{n}{8}$ gives us the claimed sets.

The probability that one of $\PivotA, \PivotB$ is in \ClusterP and the
other in $\Cluster \setminus \ClusterP$ is at least
\[
\frac{\SetCard{\ClusterP} \cdot \SetCard{\Cluster \setminus \ClusterP}}{\binom{n}{2}}
\; > \; 2 \cdot \frac{1}{16} \cdot \frac{3}{16}
\; = \; \frac{3}{128}.
\]

When this event happens, we assume w.l.o.g.~that
$\PivotA \in \ClusterP$ and
$\PivotB \in \Cluster \setminus \ClusterP$.
Then, $\ClusterP \subseteq \PartitionA \subseteq \Cluster$,
and therefore
$\frac{n}{16} \leq \SetCard{\PartitionA} \leq \frac{n}{2} \leq \frac{15n}{16}$.
Because $\SetCard{\PartitionA} \geq \frac{n}{16}$,
we also get that
$\SetCard{\PartitionB}, \SetCard{\PartitionC} \leq \frac{15n}{16}$, 
completing the proof.
\end{proof}

By Lemma~\ref{lem:separator} and the preceding discussion,
the total number of ordinal queries
for both partitioning and merging
is $O(\SetCard{\ElSet})$.
Letting \QC{n} denote the expected number of comparisons for
\QuickClustering on $n$ \elements,
we obtain the recurrence
$\QC{n}
\; \leq \; \QC{\SetCard{\PartitionA}}
         + \QC{\SetCard{\PartitionB}}
         + \QC{\SetCard{\PartitionC}}
         + O(n)$.
Because
$\SetCard{\PartitionA} + \SetCard{\PartitionB} + \SetCard{\PartitionC} = n$
and
$\max (\SetCard{\PartitionA}, \SetCard{\PartitionB}, \SetCard{\PartitionC})
\leq \frac{15}{16} \cdot n$,
the solution of this recurrence is $\QC{n} = O(n \log n)$.
Hence, we have proved the following theorem.

\begin{theorem} \label{thm:quickclustering}
Algorithm~\ref{alg:quick-sort} learns the hierarchy
using $O(n \log n)$ ordinal queries in expectation.
\end{theorem}

The query complexity of 
Algorithm~\ref{alg:quick-sort} is asymptotically optimal.
Because the hierarchical clusterings include all permutations
of \AllElements,
there are at least $2^{\Omega(n \log n)}$ different hierarchical
clusterings of $n$ \elements (see also
\cite{vikram-dasgupta:2016:interactive-hierarchical-clustering}).
More precisely, as pointed out in \cite{felsenstein:2004:phylogenies},
there are $(2n-3) \cdot (2n-5)!! = \frac{(2n-3)!}{(n-2)! \cdot 2^{n-2}}$
hierarchical clusterings.
By Stirling's approximation,
at least $n \log_2 n - O(n)$ bits of information
are necessary to uniquely identify \TreeOpt.
Since each query response reveals at most $\log_2 (3)$ bits of
information, the number of queries required is at least
$n \log_3 n - O(n)$.



\section{An \InsertionSort-like Algorithm without Noise}
\label{sec:insertion-clean}

In this section, we present and analyze 
the algorithm of Pearl and Tarsi~\cite{pearl-tarsi:1986:strucuting}.
This algorithm is reminiscent of \InsertionSort.
It is also essentially the same as an algorithm proposed by
Kannan et al.~\cite{kannan-lawler-warnow:1996:phylogeny-triplet},
as well as Tamuz et al.~%
\cite{tamuz-liu-belongie-shamir-kalai:2011:learning-kernel}
(who only considered it when the ground-truth tree is balanced and thus
has logarithmic height).
We give a self-contained analysis
of this algorithm as it is the foundation
of our main result in Section~\ref{sec:insertion-noisy},
and the analysis allows us to introduce key concepts and abstractions.

Algorithm~\ref{alg:insertion-clustering}
considers the \elements in an arbitrary order
$\AllElements = \Set{\ElS[1], \ElS[2], \ldots, \ElS[n]}$.
In each iteration $i$, \element \ElS[i] is inserted into the hierarchy
\Tree[i-1] for the preceding $i-1$ \elements.
Proposition~\ref{prop:sibling}
(whose proof follows directly from Proposition~\ref{prop:induced}) 
shows that \ElS[i] will always be inserted as a leaf sibling of
an existing (leaf or internal) node \Vertex[i], 
by inserting a new common parent of \ElS[i] and \Vertex[i].

\begin{proposition} \label{prop:sibling}
Let $\Cluster \in \SetFam[\ElS[1], \ldots, \ElS[i-1]]$ be minimal
with the property that 
$\Cluster \cup \Set{\ElS[i]} \in \SetFam[\ElS[1], \ldots, \ElS[i]]$.
(Informally speaking,
$\Cluster \subseteq \Set{\ElS[1], \ldots, \ElS[i - 1]}$ is the cluster
corresponding to the sibling of \ElS[i] in \TreeI[*]{\ElS[1], \ldots, \ElS[i]}.)
In \TreeI[*]{\ElS[1], \ldots, \ElS[i-1]},
let \Vertex[i] be the root of the subtree corresponding to the cluster
\Cluster.
Then, \TreeI[*]{\ElS[1], \ldots, \ElS[i]} is obtained from
\TreeI[*]{\ElS[1], \ldots, \ElS[i-1]} by inserting a new
internal node \VertexP as the parent of \Vertex[i],
and making the leaf \ElS[i] its other child.
\end{proposition}

\InsertAlgorithmInOneColumn{\InsertionClustering $(\ElS[1], \ldots, \ElS[n])$}%
{alg:insertion-clustering}{
\STATE{Let $\Tree[2]$ be the unique (trivial) hierarchical clustering
of \elements $\ElS[1], \ElS[2]$.}
\FOR{$i = 3, \ldots, n$}
	\STATE{Let $\Vertex[i] = \text{\FindSibling}(\Tree[i-1], \ElS[i])$.
	\label{line:vertex-search}}
	\STATE{Let \VertexP[i] be the parent of \Vertex[i]
	(or $\emptyset$ if \Vertex[i] is the root)}
	\STATE{Let \Tree[i] be the tree obtained from \Tree[i - 1]
	by adding a new vertex \VertexP with children \Vertex[i] and
	\ElS[i] and parent \VertexP[i].}
\ENDFOR
\RETURN{\Tree[n]}
}

When \TreeOpt has logarithmic height,
\FindSibling can be simply implemented as a top-down
search similar to the \Merge algorithm in Section~\ref{sec:quick-sort}.
This approach can take linear time if the tree is not balanced,
and our main focus in this section is on improving the time
to $\log n$ for arbitrary trees.

Assuming that the function \FindSibling (to be specified below)
correctly identifies \Vertex[i],
the algorithm maintains the following invariants:
\begin{enumerate}
\item Each internal node of \Tree[i] has two children;
this holds because every time an internal node is added, it is explicitly given
two children.
\item \Tree[i] is always the correct hierarchy for
$\{\ElS[1], \ElS[2], \ldots, \ElS[i]\}$,
in the sense that $\Tree[i] \equiv \TreeI[*]{\ElS[1], \ldots, \ElS[i]}$; 
this follows inductively from Proposition~\ref{prop:sibling}.
\end{enumerate}

Thus, it remains to specify an efficient implementation
of the function \FindSibling and analyze
the overall number of ordinal queries.

\FindSibling is very similar to ``\binarysearch in Trees'' 
\cite{mozes-onak-weimann:2008:tree-edge-linear,%
onak-parys:2006:tree-vertex-linear,%
2016:binary-search}
in the \emph{vertex query} model.
Recall that in the vertex query model,
an unknown node \Target is to be located in a given and known tree \Tree.
When an algorithm queries node \Vertex which is not the target,
it learns the first edge of the unique \Vertex-\Target path
(equivalently, the connected component of $\Tree \setminus \Set{\Vertex}$
in which the target lies).
The ``\binarysearch'' algorithm is based on the following observation
of Jordan~\cite{jordan:1869:assemblages}:
each tree \Tree has a separator node \Vertex with the property
that each connected component of
$\Tree \setminus \Set{\Vertex}$ contains
at most half the nodes of \Tree.
Repeatedly querying such a separator node of the remaining subtree is
easily shown to find the target \Target in a tree of $n$ nodes using at most 
$\log_2 n$ vertex queries
\cite{onak-parys:2006:tree-vertex-linear}.

\FindSibling treats the (unknown) sibling \Vertex[i]
as the target to discover,
and uses ordinal queries with internal vertices as pivots
(see Section~\ref{sec:ordinal-def}) to essentially simulate vertex queries.
It maintains a set $S$ of candidate
vertices \Vertex that have not been ruled out from being \Vertex[i].

\InsertAlgorithmInOneColumn{\FindSibling $(\Tree, \ElS)$}{alg:findsibling}{
\STATE{Let $S = V(\Tree)$ be the set of all vertices of \Tree.}
\WHILE{$\SetCard{S} > 1$}
	\STATE{Let \Vertex be a \emph{pivot vertex} minimizing
	$\max(\SetCard{\Subtree[L]{\Vertex} \cap S},
	\SetCard{\Subtree[R]{\Vertex} \cap S},
	\SetCard{\Subtree{\Vertex} \cap S})$.}
	\STATE{Make an ordinal query for \ElS with pivot \Vertex.}
	\IF{the query returns the left subtree of \Vertex}
		\STATE{Let $S = S \cap \Subtree[L]{\Vertex}$.}
	\ELSIF{the query returns the right subtree of \Vertex}
		\STATE{Let $S = S \cap \Subtree[R]{\Vertex}$.}
	\ELSE 
		\STATE{Let $S = S \cap \Subtree{\Vertex}$.}
	\ENDIF
\ENDWHILE
\RETURN{the unique vertex \Vertex in $S$.}
}

\begin{lemma} \label{lem:clean-insertion}
When all answers to ordinal queries are correct,
$\text{\FindSibling}(\Tree, \ElS)$ finds \Vertex[i] in \Tree
using at most $\log \SetCard{\Tree}$ ordinal queries.
\end{lemma}

\begin{proof}
Inductively, the induced subgraph \TreeI{S} always forms a binary tree
in which each internal node has two children.
In verifying the correctness of the algorithm,
notice first that the pivot vertex \Vertex is indeed never a leaf.
That is because for any leaf \Vertex, we would have
$\Subtree{\Vertex} \cap S = S$,
whereas for \Vertex's parent \VertexP, we would have that
$\max(\SetCard{\Subtree[L]{\VertexP} \cap S},
	\SetCard{\Subtree[R]{\VertexP} \cap S},
	\SetCard{\Subtree{\VertexP} \cap S}) 
	< \SetCard{S} = \SetCard{\Subtree{\Vertex} \cap S}$.
The fact that the update of $S$ is correct follows from
Proposition~\ref{prop:pivot-query}.



Next, we analyze the number of ordinal queries made.
By the result of Jordan \cite{jordan:1869:assemblages},
\TreeI{S} has a separator node \Vertex with the property that each of 
$\Subtree[L]{\Vertex} \cap S, \Subtree[R]{\Vertex} \cap S,
(\Subtree{\Vertex} \cap S) \setminus \Set{\Vertex}$
contains at most $\SetCard{S}/2$ nodes.
Because in the third case (\ElSL and \ElSR are closest),
the node \Vertex remains in $S$,
the size of $S$ might change to $1 + \SetCard{S}/2$ instead of $\SetCard{S}/2$.
However, notice that this can only happen when $\SetCard{S}$ is odd.
For when $\SetCard{S}$ is even, if \Subtree{\Vertex} contained
$1 + \SetCard{S}/2$ nodes from $S$,
then choosing the parent \VertexP of \Vertex instead
would strictly decrease the size of \Subtree{\VertexP},
while the size of each of the subtrees of \VertexP is bounded by $\SetCard{S}/2$.
Thus, we obtain that the new size of $S$ is at most $\Ceil{\SetCard{S}/2}$.

When $\SetCard{S} = 3$, a single query suffices to identify \Vertex[i];
the case $\SetCard{S} = 4$ cannot arise, as there is no binary tree with 4
nodes in which all internal nodes have degree 2.
Thus, the number \QC{k} of required ordinal queries satisfies the
recurrence $\QC{k} \leq 1$ for $k \leq 4$ and
$\QC{k} \leq 1 + \QC{\Ceil{k/2}}$ for $k \geq 5$.
An easy induction proof shows that $\QC{k} \leq \log k$,
so the number of ordinal queries required to find \Vertex[i] in \Tree
is at most $\log \SetCard{\Tree}$.
\end{proof}

\begin{theorem} \label{thm:clean-insertion}
When the answers to all ordinal queries are correct,
Algorithm~\ref{alg:insertion-clustering}
learns the correct hierarchical clustering of $n$ \elements
using at most $n \log_2 n$ ordinal queries.
\end{theorem}

\begin{proof}
The correctness of the algorithm follows by repeatedly applying
Proposition~\ref{prop:sibling},
and using that by Lemma~\ref{lem:clean-insertion},
each iteration uses the correct \Vertex[i].

By Lemma~\ref{lem:clean-insertion}, because \Tree[i] has $2i - 1$ nodes
($i$ \element leaves and $i - 1$ internal nodes),
inserting \element \ElS[i] into \Tree[i-1] for
$i = 3, \ldots, n$ requires at most $\log_2(2i-3)$ ordinal queries.
Thus, the total number of ordinal queries is at most  
\Eat{\begin{align*}
\sum_{i = 3}^{n} \log_2 (2i - 3)
& < \; \sum_{i = 2}^{n - 1} (1 + \log_2 i)
\; < \; n + \log_2((n - 1)!)
\; \stackrel{(*)}{\leq} \;
   n + \left(n \log_2 n - n \log_2 e + O(\log n) \right)\\
& < n \log_2 n.
\end{align*}}
\begin{align*}
\sum_{i = 3}^{n} \log_2 (2i - 3)
< \; \sum_{i = 2}^{n - 1} (1 + \log_2 i)
< n + \log_2((n - 1)!) \\
\stackrel{(*)}{\leq}
n + \left(n \log_2 n - n \log_2 e + O(\log n) \right)
< n \log_2 n.
\end{align*}
The step labeled (*) used Stirling's inequality.
\end{proof}


\section{Dealing with Noisy Feedback}
\label{sec:insertion-noisy}

In this section, we turn our attention to the noisy model:
the response to any ordinal query is correct independently
with probability $p > \Half$, and adversarially incorrect otherwise.
As discussed earlier, a very similar model has been studied in
\cite{brown-truszkowski:2011:phylogeny-quartet}
(see Section~\ref{sec:introduction}).
Our main result is the following theorem:

\begin{theorem} \label{thm:insertion-noisy}
When each query response is correct with constant probability $p > \Half$,
and adversarially incorrect with probability $1 - p$,
there is an algorithm that uses $O (n \log n + n \log (1/\Err))$ ordinal
queries, and learns the correct hierarchical clustering
with probability at least $1 - \Err$.
\end{theorem}

Since ordinal queries are only used in the call to \FindSibling
(in Line~\ref{line:vertex-search}) inside the algorithm
\InsertionClustering, it suffices to show how to find \Vertex[i]
in each iteration with high probability
using $O(\log n + \log(1/\delta))$ ordinal queries.
This is guaranteed by the following lemma, the focus of this section.

\begin{lemma} \label{lem:noisy-sibling}
Assume that each response to an ordinal query is correct
independently with probability $p > \Half$,
and adversarially incorrect with probability $1 - p$.
There exists an adaptive algorithm with the following property:
Given a hierarchical clustering \Tree
with $i$ leaves, an \element \ElS to insert,
and an error guarantee $\Err > 0$,
the algorithm finds the sibling \VertexBar of \ElS
using at most\footnote{The big-$O$ notation hides
constants depending on $p$, but not on any other parameter.}
$O(\log i + \log(1/\Err))$ ordinal queries.
\end{lemma}

\begin{emptyextraproof}{Theorem~\ref{thm:insertion-noisy}}
Using Lemma~\ref{lem:noisy-sibling},
the proof of Theorem~\ref{thm:insertion-noisy} is straightforward:
set $\Err' = \Err/n$ and apply Lemma~\ref{lem:noisy-sibling} to each
invocation of \FindSibling.
Then, by the union bound, all iterations will succeed with 
probability at least $1-\Err$.
The total number of ordinal queries will be at most
\Eat{\[
  \sum_{i=3}^n O(\log (2i-3) + \log(1/\Err'))
  \; \leq \; n \cdot O(\log n + \log (n/\Err))
  \; = \; O(n (\log n + \log(1/\Err))).
\]}
\begin{align*}
\sum_{i=3}^n O(\log (2i-3) + \log(1/\Err'))
\leq n \cdot O(\log n + \log (n/\Err))
= O(n (\log n + \log(1/\Err))). \QED
\end{align*}
\end{emptyextraproof}

\subsection{Simulating Vertex Queries with Ordinal Queries}

To avoid essentially repeating a large amount of prior analysis,
we begin by making the connection between ordinal queries
and vertex queries more explicit:
we show how to simulate one vertex query using a constant number of
ordinal queries.

To illustrate the reduction most easily, consider the case $p = 1$
in which all queries are answered correctly.
To simulate a vertex query to the root, simply make an ordinal 
query for \ElS with the root as pivot.
This one query suffices, since the response reveals the subtree that 
\ElS belongs to,
or otherwise that the root is the sibling of \ElS (if \ElS does not
belong to either of the subtrees rooted at the root's children).
To simulate a vertex query to a leaf node \Vertex,
instead make an ordinal query with the parent \VertexP of \Vertex as the pivot.
If the response points to \Vertex, then return \Vertex as the sibling;
otherwise, \ElS is not the sibling of \Vertex.

Finally, to simulate a vertex query to an internal vertex \Vertex
that is not the root,
let \VertexP be the parent of \Vertex in \Tree,
and make two ordinal queries for \ElS,
one with pivot \Vertex and the other with pivot \VertexP.
If in response to the ordinal query with pivot \Vertex,
a child of \Vertex is identified as containing \ElS in its subtree,
then this case directly corresponds to the corresponding vertex query
response at \Vertex.
If \Subtree{\Vertex} is known to contain \ElS,
then the response to the query with pivot \VertexP clarifies whether
the target is in fact \Vertex itself, or contained in
$\Subtree{\Vertex} \setminus \Set{\Vertex}$.
Thus, the responses to the two queries together contain (at least)
all the information obtained from the vertex query. 
(Notice that applying this reduction, we could have obtained the
result of Lemma~\ref{lem:clean-insertion} with a loss of a factor of 2
in the number of queries required.)

We now show how to generalize this idea to $p \in (\Half, 1)$.
The cases of the root vertex and a leaf node are handled exactly as in
the case $p = 1$; in particular, using only a single ordinal query.
To simulate a vertex query to an internal vertex \Vertex,
we choose a suitably large constant $k_p$ (determined below)
and make $k_p$ ordinal queries for \ElS with pivot \Vertex.
If a strict majority of these queries
places \ElS in one of the subtrees below \Vertex,
then that subtree is returned.
If a strict majority of the queries with pivot \Vertex places \ElS in
\Subtree{\Vertex}, then \Vertex's parent \VertexP 
is queried $k_p$ times to determine
if the desired target is equal to \Vertex or contained in
$\Subtree{\Vertex} \setminus \Set{\Vertex}$.
The former is returned as the answer if a strict majority of the responses
place \ElS in the subtree of \VertexP also containing \Vertex.
The latter is returned when a strict majority places \ElS in the other
subtree below \VertexP, or a strict majority place it in
\Subtree{\VertexP}.
Finally, in all other cases, the responses are inconclusive,
and the simulation algorithm returns an arbitrary response.
Formally, the algorithm
(which has a number of tedious case distinctions)
is given as Algorithm~\ref{alg:simulate}.

\InsertAlgorithmInOneColumn{\SimulateVertexQuery $(\Tree, \Vertex, \ElS)$}{alg:simulate}{
\IF{\Vertex is the root of \Tree}
	\STATE{Make one ordinal query for \ElS with pivot \Vertex,
	and return its result.\\
	(If $\Subtree{\Vertex} = \Set{\Vertex}$,
	the response reveals the root as the proposed answer.)}
\ELSIF{\Vertex is a leaf}
	\STATE{Let \VertexP be the parent of \Vertex in \Tree.}
	\STATE{Make one ordinal query for \ElS with pivot \VertexP.}
	\IF{the response places \ElS at \Vertex}
		\RETURN{the vertex \Vertex as a proposed answer.}
	\ELSE
	\RETURN{the direction towards \VertexP.}
\ENDIF
\ELSE
	\STATE{Let \VertexP be the parent of \Vertex
	and $k_p$ be an odd constant
	such that $1 - e^{-k_p(2p - 1)^2/2} \geq \sqrt{p}$.}
	\STATE{Make $k_p$ ordinal queries for \ElS with pivot \Vertex.}
	\IF{a strict majority of responses place \ElS in \Subtree[L]{\Vertex}}
		\RETURN{the direction towards the left child of \Vertex.}
	\ELSIF{a strict majority of responses place \ElS in \Subtree[R]{\Vertex}}
		\RETURN{the direction towards the right child of \Vertex.}
	\ELSIF{a strict majority of responses place \ElS in \Subtree{\Vertex}}
		\STATE{Make $k_p$ ordinal queries for \ElS with pivot \VertexP.}
		\IF{a strict majority of responses place \ElS in the subtree of
			\VertexP containing \Vertex}
			\RETURN{the vertex \Vertex as a proposed answer.}
		\ELSIF{a strict majority of responses place \ElS in
		the other subtree of \VertexP or \Subtree{\VertexP}}
			\RETURN{the direction towards \VertexP.}
		\ELSE
			\RETURN{an arbitrary response.}
		\ENDIF
	\ELSE
		\RETURN{an arbitrary response.}
\ENDIF
\ENDIF
}

\begin{lemma} \label{lem:simulation}
If each ordinal query is answered correctly independently
with probability $p$, then (for a suitable choice of $k_p$, specified below),
\SimulateVertexQuery $(\Tree, \Vertex, \ElS)$ returns the correct
response (either \Vertex or the direction in which to search)
to the vertex query for \Vertex, with probability at least $p$.
\end{lemma}

\begin{proof}
By the additive Hoeffding Bound \cite{hoeffding:1963:probability},
the probability that a majority out of $k$ query responses is correct
for every $p > \Half$
is at least $q(k) := 1- e^{-k(2p - 1)^2/2}$.
For odd $k$, any majority is strict.
For any $p > \Half$, we have that $e^{-(2p - 1)^2/2} < 1$,
so $q(k) \to 1$ as $k \to \infty$.
In particular, $q(k) \geq \sqrt{p}$ for $k$ sufficiently large.
Fixing such a $k_p$, the probability that a strict majority
of the queries both with pivots \Vertex, \VertexP are correct
is at least $(\sqrt{p})^2 = p$.
Whenever strict majorities of both sets of ordinal queries are correct,
the output of \SimulateVertexQuery is correct, completing the proof. 
\end{proof}

\subsection{Reanalyzing Binary Search with Noisy Vertex Queries}

In light of Lemma~\ref{lem:simulation},
for the purpose of analysis in this section,
we can ignore ordinal queries,
and instead focus on finding \ElS's sibling \VertexBar using vertex queries:
in response to querying \Vertex, the algorithm is told either that
$\Vertex = \VertexBar$ is the correct sibling of \ElS,
or is given a subtree
(rooted at one of the children of \Vertex,
or obtained by removing \Vertex and its subtrees)
that must contain \ElS.
This information is correct with probability (at least) $p$.

\cite{2016:binary-search} gave an algorithm for \binarysearch on
undirected graphs with noisy responses that solves this problem with
probability at least $1 - \Err$, using at most $O(\log n + \log^2 (1/\Err))$
vertex queries \cite[Theorem 8]{2016:binary-search}.
Since we need $\Err = O(1/n)$ to be able to take a Union Bound,
applying this result directly would only give us an overall guarantee of
$O(n \log^2 n)$.
Hence, in this section, we show how to combine algorithms of 
\cite[Algorithm 2]{2016:binary-search} and of
Feige et al.~\cite{feige-raghavan-peleg-upfal:1994:noisy}
(see also related algorithms for noisy \binarysearch on the line
\cite{BenOr-hassidim:2008:noisy-binary-search,%
karp-kleinberg:2007:noisy}) 
to obtain an algorithm finding \VertexBar in $O(\log n + \log (1/\Err))$ queries.
In contrast to \cite{2016:binary-search},
the dependence of this new algorithm's number of queries on $p$ is not
information-theoretically optimal, but it obtains an improved
dependence on \Err in return.

Algorithm~2 of \cite{2016:binary-search} is a
variant of a multiplicative weights algorithm which starts with
a candidate set $S$ of vertices, and within essentially $O(\log |S|)$ rounds
reduces it to only $O(\log |S|)$ ``likely'' remaining candidates. 
The overall algorithm in \cite{2016:binary-search} works by first
reducing the candidate set of \emph{all} nodes to just
$O(\log n)$ nodes, then reducing that further to $O(\log \log n)$
nodes with a second invocation of the multiplicative weights algorithm,
and finally querying each remaining candidate frequently enough to
identify the correct node with sufficiently high probability.
(Dependencies on \Err are omitted in our high-level overview.)
As part of our algorithm here, we just use a single invocation of
the Multiplicative Weights algorithm
(Algorithm 2 of \cite{2016:binary-search}),
with parameters modified to obtain a better dependence on \Err at the
cost of a worse dependence on $p$. 
The performance of of this algorithm (used as a Black Box here)
is summarized by Lemma~\ref{lem:multiplicative-weights},
which is obtained from Lemma~5 of \cite{2016:binary-search} by setting
$\lambda = \frac{1 + p \log p + (1-p) \log(1-p)}{2 \log(p / (1-p))}$.

\begin{lemma} \label{lem:multiplicative-weights}
There exists an algorithm which uses at most $O(\log n + \log(1/\Err))$ 
vertex queries and returns a set $S$ of $O(\log n + \log(1/\Err))$ nodes
such that the unknown target is contained in $S$
with probability at least $1-\Err$.
\end{lemma}

Once the number of candidate targets has been reduced to $O(\log n)$,
due to the $\log(1/\Err)$ dependence,
a further reduction using the techniques of \cite{2016:binary-search}
seems difficult.
However, at this point, we can apply a modification of a beautiful algorithm
of Feige et al.~\cite{feige-raghavan-peleg-upfal:1994:noisy}.
Translated into our nomenclature, Lemma~3.1 of 
\cite{feige-raghavan-peleg-upfal:1994:noisy} about noisy \binarysearch
can be phrased as follows:

\begin{lemma}[Lemma~3.1 of \cite{feige-raghavan-peleg-upfal:1994:noisy}, restated]
\label{lem:feige-restated}
Let \Tree be a binary tree, and assume that one of the leaves
of \Tree is the unknown target.
In response to a vertex query, assume that the algorithm is given the
correct answer with probability $p$,
and an adversarially incorrect one otherwise.
Then, there is an algorithm which finds the correct target
with probability at least $1-\Err$, 
using $O(\Diameter + \log (1/\Err))$ queries
where \Diameter is the diameter of \Tree.
\end{lemma}

Notice two minor inconveniences in applying
Lemma~\ref{lem:feige-restated} for our purposes:
\begin{enumerate}
\item In the specific application of \cite{feige-raghavan-peleg-upfal:1994:noisy},
the binary tree is complete and thus, \Diameter is logarithmic
in size of the tree.
In our application, on the other hand,
after invocation of the Multiplicative Weights algorithm,
the number of candidates (and hence diameter of the tree)
is bounded by $O(\log n)$.
\item \cite{feige-raghavan-peleg-upfal:1994:noisy}
assume that the target is always a leaf node,
whereas in our application, the unknown sibling \VertexBar may be
an internal node. It is very straightforward to modify the algorithm of
Feige et al.~to handle non-leaf targets.
\end{enumerate}

For completeness, we present the modified algorithm and corresponding
analysis in Appendix~\ref{sec:feige}.
Its guarantee is summarized by the following lemma:

\begin{lemma}\label{lem:feige}
Assume that each query response is correct
independently with probability $p > \Half$,
and adversarially incorrect with probability $1 - p$.
There exists an adaptive algorithm with the following property:
Given a tree \Tree of diameter \Diameter,
and error parameter $\Err > 0$,
the algorithm finds the target with probability at least $1-\Err$,
using at most $O(\Diameter + \log(1/\Err))$ vertex queries.
\end{lemma}

Using Lemmas~\ref{lem:multiplicative-weights} and \ref{lem:feige},
it is easy to see how to find a target in $O(\log n + \log(1/\Err))$ queries.
Set $\Err' = \Err/2$ and apply Lemma~\ref{lem:multiplicative-weights}.
Within $O(\log n + \log(1/\Err')) = O(\log n + \log(1/\Err))$ iterations,
we obtain a set $S$ of at most 
$O(\log n + \log(1/\Err))$ nodes that contains the target
with probability at least $1-\Err'$.
Then run the algorithm of Lemma~\ref{lem:feige} with parameter $\Err'$
on the tree \TreeI{S}.
If the target was in $S$, the algorithm will find it with probability
at least $1-\Err'$.
Since the height of \TreeI{S} is at most
$\SetCard{S} = O(\log n + \log(1/\Err))$, the number of queries
required in this phase is at most $O(\log n + \log(1/\Err))$.
Thus, the total number of queries is $O(\log n + \log(1/\Err))$,
and the failure probability is at most $2\Err' = \Err$.


\section{Conclusion and Open Problems}\label{sec:conclusion}
\label{sec:generalized-hierarchical-clustering}

We presented adaptive algorithms for learning a hierarchical clustering
from $O(n \log n)$ ordinal queries when query responses are noisy.
In the absence of noise,
there is an upper bound of $n \log_2 n$ (Section~\ref{sec:insertion-clean}),
and we gave a lower bound of $n \log_3 n - O(n)$.
An immediate question is whether the constant in $n \log_2 n$ is
correct in the noiseless case, or could be improved.

For all our results, we assumed that the tree \Tree capturing
the hierarchical clustering is \emph{binary}.
When nodes can have larger degrees,
our analysis does not continue to hold.
One problem is that for an ordinal query, another response of \Null
may be required: when $\ElS, \ElSP, \ElSPP$ are leaves in three
distinct subtrees of the same internal node \Vertex, 
no pair of them is closer to each other than the third.
This issue also manifests itself in a higher query complexity,
as pointed out by~\cite{benson-kumar-tomkins:2016:nested-IIA}:
when \Tree is a star in which one leaf has been replaced with
an internal node with two new leaves, and the assignment of \elements
to leaves is uniformly random,
it is not difficult to show that
any non-adaptive or adaptive algorithm (even randomized)
requires $\Omega(n^2)$ ordinal queries to learn \Tree.

When the maximum degree \MaxChildren of \Tree is not a constant,
Pearl and Tarsi~\cite{pearl-tarsi:1986:strucuting}
pointed out that their \InsertionSort-based algorithm
can be generalized to obtain an upper bound of
$O(\MaxChildren \cdot n \log n)$ on the number of queries.
However, this generalized algorithm is not optimal:
in the context of discovering phylogeny trees,
Brodal et al.~%
\cite{brodal-fagerberg-pedersen-ostlin:2001:phylogeny-higher-degree}
present an algorithm that uses $O(\frac{d}{\log d} \log n)$ queries.
Their algorithm can be rederived
by adapting an algorithm of \cite{2016:binary-search}
for the \emph{edge query model}
(using techniques similar to the ones
in Section~\ref{sec:insertion-clean})
to insert a new \element \ElS into \Tree using
$O(\frac{d}{\log d} \cdot \log n)$ ordinal queries.
An interesting open question is whether
there is a matching lower bound
of $\Omega(\frac{d}{\log d} \cdot n \log n)$
for every \MaxChildren.
For $\MaxChildren = 2$ (and more generally $\MaxChildren = O(1)$),
our counting argument from Section~\ref{sec:quick-sort}
gives a matching lower bound
up to a constant factor, and for $\MaxChildren = n-1$
(and more generally $\MaxChildren = \Omega(n)$),
the example mentioned in the preceding paragraph does.
We do not have a matching lower bound for intermediate \MaxChildren.
Another natural question is whether the upper bound of
$O(\frac{\MaxChildren}{\log \MaxChildren} \cdot n \log n)$ ordinal queries
can be achieved in the noisy model.

A different interpretation of our result
is that we show how to learn
a (very restricted) ultrametric on the \elements in \AllElements
using ordinal queries.
This raises the natural question of how well more general metrics can
be learned using ordinal queries.
Indeed, as mentioned earlier,
some past work has considered this question,
with some assumption on the geometry of the metric
(e.g.,~\cite{jain-jamieson-nowak:2016:ordinal-embedding}).
Naturally, even for very simple metrics (such as points on a line),
we cannot hope to learn all distances accurately:
for example, if one \element is further from all other \elements
than any other pair is from each other, no algorithm can learn its distance
to other \elements to within any approximation guarantee.
Is there a natural ordinal approximation to
a (not necessarily Euclidean) metric that can be
obtained efficiently?
Could techniques such as
\cite{bartal:1996:probabilistic-approximations,%
bartal:1998:tree-metrics,%
fakcharoenphol-rao-talwar:2004:tree-metrics}
of probabilistic tree embeddings be useful in this context?


\subsubsection*{Acknowledgments}
Work support by NSF grant 1619458.
We would like to thank Daniel Hsu
for pointing us to \cite{pearl-tarsi:1986:strucuting},
and Dan Brown and Jakub Truszkowski
for informing us about
\cite{brown-truszkowski:2011:phylogeny-quartet,
brown-truszkowski:2011:phylogeny-quartet-practical}.
Last but not least, we thank Sanjoy Dasgupta,
Shanghua Teng, and anonymous reviewers
for useful feedback and suggestions.


\bibliographystyle{abbrv}
\bibliography{../davids-bibliography/names,../davids-bibliography/conferences,../bibliography/references}
\appendix



\section{Proof of Theorem~\ref{thm:hardness-non-adaptive}}
\label{sec:non-adaptive-lower-bound}

Here, we prove Theorem~\ref{thm:hardness-non-adaptive},
restated for convenience.

\begin{rtheorem}{Theorem}{\ref{thm:hardness-non-adaptive}}
Any non-adaptive algorithm requires $\Omega(n^3)$ ordinal queries
to learn the hierarchical clustering over $n$ \elements
(even in the absence of noise).
\end{rtheorem}

\begin{proof}
A (possibly randomized) non-adaptive algorithm poses \NumQuery queries
and receives their responses all at once.
Let $n$ be a power of $2$ and generate a ground truth
hierarchical clustering \TreeOpt as a full binary tree
in which the leaves form
a uniformly random permutation of the \elements. 
For every constant $\Err > 0$,
we prove that if the algorithm uses fewer than
$n (n - 1)(n - 2) (1 - \Err) / 24 = \Omega(n^3)$ ordinal queries,
then it fails to uniquely identify the clustering
with probability at least $1 - \Err$.

\begin{figure}[htb]
\begin{center}
\includegraphics{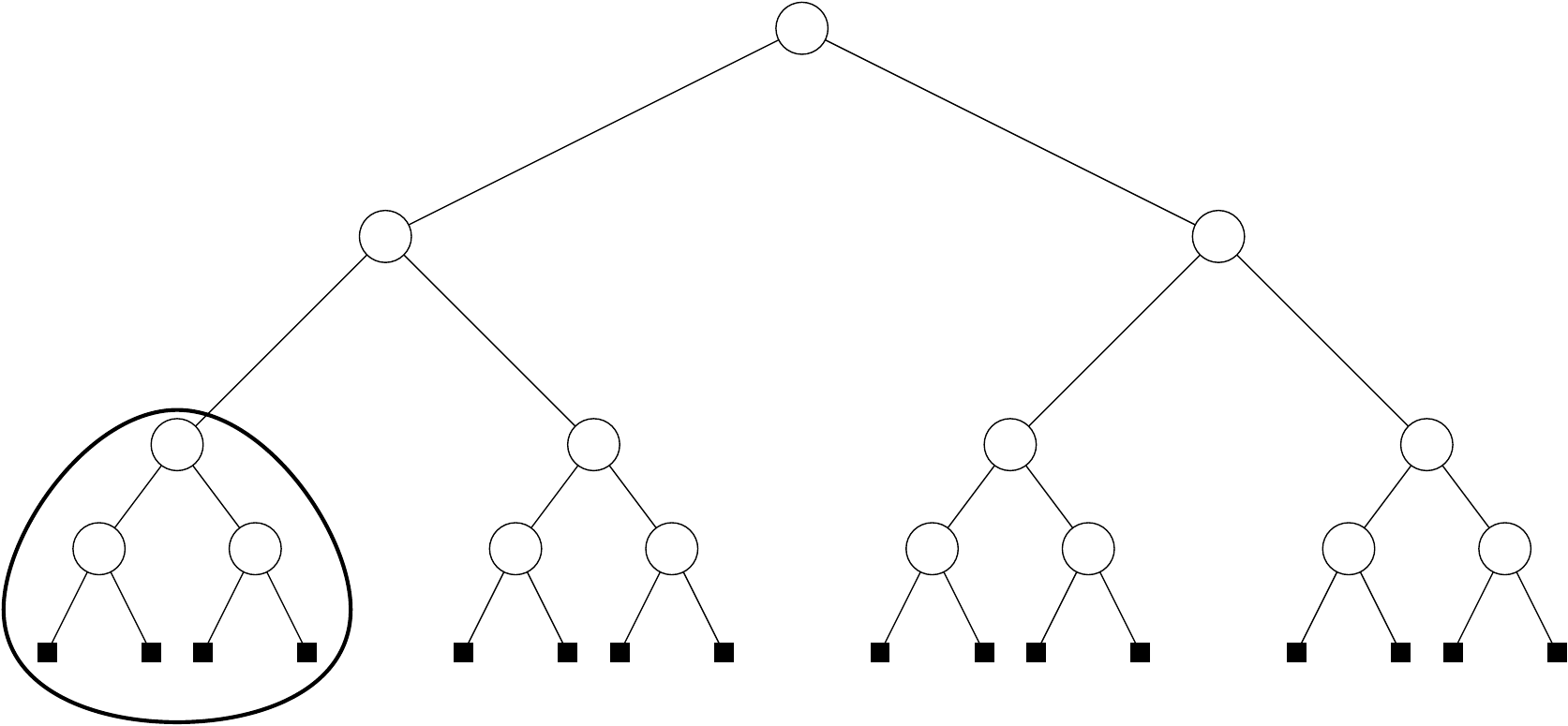}
\end{center}
\caption{A full binary tree and one of its clusters of size 4. \label{fig:complete-perfect-four}}
\end{figure}

Because \TreeOpt is a full binary tree,
there are exactly $n/4$ disjoint clusters of size 4 each
(Figure~\ref{fig:complete-perfect-four}).
Let \Cluster be one of them.
The 4 \elements in \Cluster are partitioned into two clusters
of size 2 each.
The algorithm can learn this internal structure of \Cluster
only if one of its queries consists of 3 (out of the 4) \elements in \Cluster.
Because the \elements are uniformly randomly shuffled,
the probability that the 3 \elements of any one query all lie in
\Cluster is $\frac{24}{n (n - 1)(n - 2)}$.

Because there are $n/4$ such clusters \Cluster, 
if the algorithm queries \NumQuery triplets,
in expectation, it learns the internal structure of at most
$\NumQuery \cdot \frac{n}{4} \cdot \frac{24}{n (n - 1)(n - 2)}
= \frac{6\NumQuery}{(n-1)(n-2)}$
clusters of size 4.
If the algorithm succeeds with probability at least $1 - \Err$,
then it learns the internal structure of at least
$(1 - \Err)\frac{n}{4}$ size-4 clusters in expectation.
Thus, $\frac{6\NumQuery}{(n-1)(n-2)} \geq (1 - \Err)\frac{n}{4}$
which implies $\NumQuery \geq (1 - \Err) n (n - 1) (n - 2) / 24$.
\end{proof}


\section{Proof of Lemma~\ref{lem:feige}}
\label{sec:feige}

In this section, we provide the proof of Lemma~\ref{lem:feige}.
We restate the lemma here for convenience.

\begin{rtheorem}{Lemma}{\ref{lem:feige}}
Assume that each query response is correct
independently with probability $p > \Half$,
and adversarially incorrect with probability $1 - p$.
There exists an adaptive algorithm with the following property:
Given a tree \Tree of diameter \Diameter, and error parameter $\Err > 0$,
the algorithm finds the target with probability at least $1 - \Err$,
using at most $O(\Diameter + \log(1/\Err))$ vertex queries.
\end{rtheorem}

Our algorithm is a very minor modification of an algorithm 
of Feige et al.~\cite{feige-raghavan-peleg-upfal:1994:noisy}.
The algorithm performs a walk on the tree \Tree:
in each iteration, a vertex \Vertex is queried,
and if the response points to a neighbor \VertexP,
the algorithm next ``moves to'' \VertexP.
However, with every response confirming \Vertex as the target,
the algorithm moves ``deeper'' into \Vertex
(by incrementing a counter):
subsequently, responses pointing to neighbors of \Vertex will only
move the walk one step ``shallower'' in \Vertex,
and the walk will only leave \Vertex once the counter has been reduced
down to 0.


More formally, each vertex \Vertex in \Tree has a counter
\Counter{\Vertex}.
All the counters are initially $0$,
and at any point of the algorithm, at most one counter will have a
non-zero (positive) value.

In the \Ordinal{i} iteration, the algorithm queries a vertex \Query[i].
The response to the query is either the node \Query[i] itself
or one of its neighbors in \Tree.
If the query response is \Query[i], i.e.,
the noisy response proposes \Query[i] as the target,
then the algorithm increments \Counter{\Query[i]}
and will next query the same node $\Query[i+1] = \Query[i]$.
If the query response is a node other than \Query[i],
the algorithm's action depends on \Counter{q_i}.
If $\Counter{\Query[i]} > 0$,
the algorithm decrements \Counter{\Query[i]}
and will next query the same node $\Query[i+1] = \Query[i]$.
Finally, if $\Counter{q_i} = 0$, then $\Query[i + 1]$ is simply the
vertex given in response to the query.
The algorithm is given formally as Algorithm~\ref{alg:feige}.

\InsertAlgorithmInOneColumn{\TreeWalk $(\Tree, \Err)$}%
{alg:feige}{
\STATE{$\Counter{\Vertex} \AlgAssign 0$ for every vertex \Vertex in \Tree.} 
\STATE{Let \Query be an arbitrary vertex in \Tree.}
\FOR{$i = 1, \ldots
\max(\frac{2 (D+1)}{2p - 1}, \frac{8 \ln(1/\Err)}{(2p - 1)^2})$}
\STATE{Query the node \Query, and let \Response be the noisy response.}
\IF{$\Response = \Query$}
\STATE{$\Counter{\Query} = \Counter{\Query} + 1$.}
\ELSIF{$\Counter{\Query} > 0$}
\STATE{$\Counter{\Query} = \Counter{\Query} - 1$.}
\ELSE
\STATE{$\Query \AlgAssign \Response$.}
\ENDIF
\ENDFOR
\RETURN{\Query}
}

We use \Counter[i]{\Vertex} to denote 
the value of the counter \Counter{\Vertex} after the \Ordinal{i} iteration.
Let \Dis{\Vertex}{\VertexP} denote the distance between \Vertex and \VertexP
in \Tree, and define the \emph{potential value} of \Vertex in iteration $i$
as
\begin{align*}
\Potential[\Vertex]{i}
& = \Dis{\Vertex}{\Query[i]} - \Counter[i]{\Vertex}
+ \sum_{\VertexP \neq \Vertex} \Counter[i]{\VertexP}.
\end{align*}

The following lemma states that the
potential of the (unknown) target node must decrease
with each correct response to a query
(and may increase with incorrect responses).

\begin{lemma} \label{lem:potential}
Let \Target be the (unknown) target.
For every $i$, if the response to the \Ordinal{i} query is correct,
then $\Potential[\Target]{i + 1} = \Potential[\Target]{i} - 1$;
if the response is incorrect, then
$\Potential[\Target]{i + 1} \leq \Potential[\Target]{i} + 1$.
\end{lemma}

\begin{emptyproof}
We distinguish two cases, based on the node queried in the \Ordinal{i} iteration.
\begin{itemize}
\item If $\Query[i] = \Target$, then a correct response increments 
\Target's counter, and thus decreases \Target's potential by 1.
An incorrect response either decrements \Target's counter
or --- if the counter was 0 --- moves the query node to one of
\Target's neighbors.
In either case, \Target's potential is only increased by 1.
\item If $\Query[i] \neq \Target$, then a correct response points
towards \Target. 
If $\Counter[i]{\Query[i]} > 0$, then \Counter{\Query[i]} is decremented;
if $\Counter[i]{\Query[i]} = 0$,
then the new query $\Query[i+1] = \Response[i]$ is made to a node one
step closer to \Target.
Therefore, \Target's potential is decreased by 1.

An incorrect response could increment the counter of \Query[i]
(if it points to \Query[i]),
decrement the counter of \Query[i]
(if it points to an incorrect neighbor of \Query[i] and
$\Counter[i]{\Query[i]} > 0$),
or move the next query $\Query[i+1] = \Response[i]$ to an incorrect
neighbor of \Query[i] that is one step further away from \Target.
In all three cases, the potential of \Target increases or decreases by 1;
in particular, it at most increases by 1.\QED
\end{itemize}
\end{emptyproof}

\begin{extraproof}{Lemma~\ref{lem:feige}}
Algorithm~\ref{alg:feige} runs for \NumQuery iterations.
In expectation, a $p$ fraction of the query responses are correct.
Let $\epsilon = \frac{2p - 1}{4} < p$.
Using Hoeffding's inequality \cite{hoeffding:1963:probability},
with probability at least $1 - e^{- 2 \NumQuery \epsilon^2}$,
at least $(p - \epsilon) \NumQuery$ query responses are correct.
In that case, by Lemma~\ref{lem:potential},
\begin{align}
\Potential[\Target]{\NumQuery}
& \leq \Potential[\Target]{0} - (p - \epsilon) \NumQuery +
(1- p + \epsilon) \NumQuery \nonumber
\\ & = \Dis{\Target}{\Query[0]} - (2p - 1 - 2\epsilon) \NumQuery
\leq D - \frac{2p-1}{2} \NumQuery. \label{equ:potential}
\end{align}

Because $\NumQuery \geq \frac{\ln(1/\Err)}{2 \epsilon^2}$,
the success probability is at least
$1 - e^{- 2 \epsilon^2 \NumQuery} \geq 1 - \Err$.
And because $\NumQuery \geq \frac{2 (D+1)}{2p - 1}$,
Inequality~\ref{equ:potential} implies that
$\Potential[\Target]{\NumQuery} \leq D - (D+1) < 0$.
Therefore, by definition of the potential value,
$\Counter[\NumQuery]{\Target} > 0$.
The only node that could have a positive counter
in any iteration $i$ is \Query[i];
therefore, the returned node must be \Target in the high-probability case.
\end{extraproof}


\end{document}